\documentclass[9pt,technote]{IEEEtran}
\usepackage{amssymb}
\usepackage{amsmath}
\usepackage{graphicx}
\usepackage{cite}

\begin{document}

\title{Statistical Performance Analysis of MDL Source Enumeration
in Array Processing}
\author{F.~Haddadi*,
        M.~Malek Mohammadi,
        M.~M.~Nayebi,
        and~M.~R.~Aref%
\thanks{Authors are with the Department of Electrical Engineering,
Sharif University of Technology, Tehran, Iran (e-mails:
farzanhaddadi@yahoo.com, m.rezamm@ieee.org, Nayebi@sharif.ir, and
Aref@sharif.ir). This work was supported in part by the Advanced
Communication Research Institute (ACRI), Sharif University of
Technology, and in part by Iran Telecommunication Research Center
(ITRC) under counteract number T/500/20613.}}

\maketitle

\begin{abstract}
In this correspondence, we focus on the performance analysis of
the widely-used minimum description length (MDL) source
enumeration technique in array processing. Unfortunately,
available theoretical analysis exhibit deviation from the
simulation results. We present an accurate and insightful
performance analysis for the probability of missed detection. We
also show that the statistical performance of the MDL is
approximately the same under both deterministic and stochastic
signal models. Simulation results show the superiority of the
proposed analysis over available results.
\end{abstract}


\begin{center}
\bfseries EDICS Category: SAM-PERF, SAM-SDET
\end{center}

\IEEEpeerreviewmaketitle

\newtheorem{lemma}{Lemma}
\renewcommand{\QED}{\QEDopen}

\section{Introduction and Preliminaries}

MDL \cite{wax}, is one of the most successful methods for
determining the number of present signals in array processing and channel order detection \cite{liavas-ch}. MDL is a low complexity information theoretic criteria which does not need any subjective threshold setting usual in detection theoretic criteria. Other statistical properties, specially its asymptotic
consistency \cite{wax}, makes it a favorable choice for source
enumeration. Unfortunately, only few approximate finite-sample
performance analysis are available on the MDL method \cite{wang,
kavehCOR, zhang, xu-kaveh, fishler, liavas-perf}. In \cite{wang}, a
simple asymptotic statistical model for the eigenvalues of the
sample correlation matrix was used. Unfortunately, the theoretical
results showed persistent bias from the simulation results
\cite{kavehCOR}.

The next work \cite{zhang}, gives a computational approach for
calculation of the probability of false alarm $p_{fa}$. In
calculating the probability of missed detection $p_m$, the same
inaccurate statistical model is used as in \cite{wang}. In
\cite{xu-kaveh}, instead of exact performance estimation,
theoretical bounds for performance were presented. A qualitative
performance evaluation in terms of gap between noise and signal
eigenvalues and also the dispersion of each group is given in
\cite{liavas-perf}. In a recent work \cite{fishler}, a significantly
different approach was used. Our simulation results show improved
results of \cite{fishler} in comparison with \cite{wang}. The
performance analysis was generalized to the non-Gaussian signals
while it was shown that the results reduce to the results of
\cite{zhang, xu-kaveh} in Gaussian signals. We will show that the
same modelling errors have degraded the analysis in \cite{fishler}
as in \cite{wang, kavehCOR, zhang, xu-kaveh}.

In this correspondence, we use an approach very similar to
\cite{wang, kavehCOR, zhang} to estimate $p_m$, including in the
analysis the finite sample $\mathcal{O} (n^{-1})$ biases of the
eigenvalues. The noise subspace eigenvalue spread is taken into
account which prevents the signal subspace eigenvalues to approach
$\sigma^2$, the noise variance. The bias of the noise power
estimator in MDL is calculated to get excellent match between
theoretical and simulation results. We will not calculate $p_{fa}$
which is negligible.

In the previous works, only the case of stochastic signal has been
considered. Here, we use a perturbation analysis to calculate
biases and variances of the eigenvalues under deterministic
signal, too. Using these results, we show that the performance of
source enumeration methods are approximately the same in both
stochastic and deterministic signal models. This is a natural
complementary result for the known fact that the performance of
the DOA (Direction of Arrival) estimation methods in array
processing is the same under stochastic and deterministic signal
models \cite{ottersten}.

From a sensor array of $L$ elements, $n$ observations $\boldsymbol
x_i \in \mathbb{C}^{L \times 1} ,  i=1,\ldots,n$ is made, which is
a linear transformation of $d<L$ source signals $\boldsymbol s_i
\in \mathbb{C}^{\, d \times 1}$, plus noise $\boldsymbol \nu_i \in
\mathbb{C}^{L \! \times 1}$
\begin{equation}
\label{eq:model} \boldsymbol x_i=\boldsymbol A(\boldsymbol \theta)
\boldsymbol s_i + \boldsymbol \nu_i
\end{equation}
where $\boldsymbol A \in \mathbb{C}^{L \! \times d}$, the steering
matrix, is composed of $d$ linearly independent column vectors of
array response $\boldsymbol a(\theta_k) , k=1,\ldots,d$. Let
$\boldsymbol X \triangleq [\boldsymbol x_1 , \ldots , \boldsymbol
x_n]$ and $\boldsymbol S$ and $\boldsymbol V$ be defined in the
same way. Signal and noise are assumed to be iid and uncorrelated
random variables. A compact form for the model will be
\begin{equation}
\label{eq:model_matrix} \boldsymbol X=\boldsymbol A(\boldsymbol
\theta) \boldsymbol S + \boldsymbol V.
\end{equation}
Noise is assumed to be circular Gaussian. Signal can be modelled
either as a zero-mean circular Gaussian random sequence or an
unknown deterministic sequence. The distribution of $\boldsymbol
x$ will be as $\mathcal{N}(\boldsymbol 0, \boldsymbol {APA}^
\textrm{H} + \sigma^2 \boldsymbol I)$ where $\boldsymbol P = E
(\boldsymbol {ss}^ \textrm{H})$ in the \emph{stochastic} signal
model, and as $\mathcal{N}(\boldsymbol {As}\, , \sigma^2
\boldsymbol I)$ in the \emph{deterministic} signal model.

To estimate the number of present signals $d$, eigenvalues of the
correlation matrix $\boldsymbol R = n^{-1} E (\boldsymbol
{XX}^\textrm{H})$ are used. Note that $\boldsymbol R_{det} =
n^{-1} \boldsymbol {ASS\,}^\textrm{H} \! \boldsymbol A^\textrm{H}
+ \sigma^2 \boldsymbol I$ and $\boldsymbol R_{sto} = \boldsymbol
{APA}^\textrm{H} + \sigma^2 \boldsymbol I$. The eigendecomposition
of the correlation matrix is
\begin{equation}
\label{eq:evR} \boldsymbol {Rv}_i = \lambda_i \boldsymbol v_i
\end{equation}
and we have $\lambda_1 > \cdots > \lambda_d
> \lambda_{d+1}= \cdots =\lambda_{L}= \sigma^2$.
Source enumeration methods are based on a spherity test on the
sample correlation matrix defined as
\begin{equation}
\label{eq:sample_R}  \hat {\boldsymbol R} = \frac{1}{n}
\sum_{i\,=1}^{n} \boldsymbol x_i \boldsymbol x_i^ \textrm{H}.
\end{equation}
Eigendecomposition of $\hat {\boldsymbol R}$ is defined as
$\boldsymbol {\hat {R}w}_i = l_i \boldsymbol w_i$ in which $l_1 >
l_2 > \cdots > l_L$. The MDL estimator of $d$ is the minimizer of
the following criterion
\begin{equation}
\label{eq:MDL} \Lambda(d,L,n) = n(L-d \, ) \log \bigg(
\frac{a_d}{g_d} \bigg) +  \frac{1}{2} \, d(2L-d) \log(n)
\end{equation}
where
\begin{equation}
\label{eq:a_d}   a_d  \triangleq  \frac{1}{L-d} \sum_{i=d+1}^{L}
l_i
\end{equation}
\begin{equation}
\label{eq:g_d} g_d  \triangleq  \prod_{i=d+1}^{L} l_i^{1/(L-d)}
\end{equation}
The first term in \eqref{eq:MDL} is the generalized likelihood
ratio for the test of spherity and the second term is a penalty
function preventing over-modelling.

\section{Statistical Properties of Eigenvalues}
\label{sec:stat}

\subsection{Signal Eigenvalues}
\label{subsec:stat_signal}

First of all, we derive a result useful for statistical
characterization of the signal eigenvalues in the deterministic
signal model. Let $\boldsymbol x_i \in \mathbb{C}^{\, L \! \times
1} \, , \, i=1,\ldots,n$ be i.i.d. observations and $\boldsymbol
x_i \sim \mathcal{N} ( \boldsymbol 0 , \boldsymbol \Sigma \, )$.
Note that $\textrm{vec} \, (\boldsymbol X) \sim
\mathcal{N}(\boldsymbol 0 \, , \boldsymbol I_n \otimes \boldsymbol
\Sigma \, )$, where $\otimes$ is the Kronecker product and
$\textrm{vec} (\boldsymbol X)$ is the vectorizing operator
stacking columns of $\boldsymbol x$ in a single column vector. Let
$\boldsymbol \alpha ,\boldsymbol \beta , \boldsymbol \gamma ,
\boldsymbol \zeta \in \mathbb{C}^{L \! \times 1}$ be constant
vectors. The Brillinger result states that \cite[p.
114]{brillinger}:
\begin{equation}
\label{eq:brill} \textrm{Cov} (\boldsymbol \alpha ^ \textrm{H}
\hat {\boldsymbol R}\, \boldsymbol \beta \, , \boldsymbol \gamma ^
\textrm{H} \hat {\boldsymbol R} \, \boldsymbol \zeta ) =
n^{-1}(\boldsymbol \alpha ^ \textrm{H} \boldsymbol \Sigma \,
\boldsymbol \gamma )(  \boldsymbol \zeta^ \textrm{H} \boldsymbol
\Sigma \, \boldsymbol \beta).
\end{equation}
We generalize the Brillinger result to the nonzero-mean case. To
the best of our knowledge the following result is new to the
literature.

\begin{lemma}
\label{lem:Bril_nc} Let $\textrm{vec} (\boldsymbol Y) \sim
\mathcal{N}(\textrm{vec} (\boldsymbol \mu) \, , \boldsymbol I_n
\otimes \boldsymbol \Sigma \,)$, where $\boldsymbol \mu \triangleq
[\, \boldsymbol \mu_1 ,\ldots, \boldsymbol \mu_n]$ and
$\boldsymbol Y \triangleq [\, \boldsymbol y_1 ,\ldots, \boldsymbol
y_n]$. Then for $\hat {\boldsymbol R} = n^{-1} \boldsymbol {YY}^
\textrm{H}$ and constant vectors $\boldsymbol \alpha ,\boldsymbol
\beta , \boldsymbol \gamma , \boldsymbol \zeta \in \mathbb{C}^{L
\! \times 1}$, we will have
\begin{eqnarray}
\label{eq:bril_nc} c \triangleq \textrm{Cov} (\boldsymbol \alpha ^
\textrm{H} \hat {\boldsymbol R}\, \boldsymbol \beta \, ,
\boldsymbol \gamma ^ \textrm{H} \hat {\boldsymbol R}\, \boldsymbol
\zeta ) = n^{-1}(\boldsymbol \alpha ^ \textrm{H} \boldsymbol
\Sigma \, \boldsymbol \gamma )(  \boldsymbol
\zeta^ \textrm{H} \boldsymbol \Sigma \, \boldsymbol \beta) \nonumber \\
{+} \: n^{-2}(\boldsymbol \alpha ^ \textrm{H} \boldsymbol
{\mu\mu}^ \textrm{H} \, \boldsymbol \gamma )( \boldsymbol \zeta^
\textrm{H} \boldsymbol \Sigma \, \boldsymbol \beta) \nonumber \\ +
\: n^{-2}(\boldsymbol \alpha ^ \textrm{H} \boldsymbol \Sigma \,
\boldsymbol \gamma )(  \boldsymbol \zeta^ \textrm{H} \boldsymbol
{\mu\mu}^ \textrm{H} \, \boldsymbol \beta)
\end{eqnarray}
\end{lemma}

\begin{proof}
See Appendix \ref{app:bril_nc}.
\end{proof}

We first briefly state useful available results.

\newtheorem{theorem}{Theorem}
\begin{theorem}
\label{th:Girshick} Let $\textrm{vec} \, (\boldsymbol X) \sim
\mathcal{N}(\boldsymbol 0 \, , \boldsymbol I_n \otimes \boldsymbol
\Sigma \, )$. Then the signal eigenvalues of $\hat{ \boldsymbol
R}$ in the asymptotic region of $n \gg 1$ has limiting Gaussian
distribution and we have \cite{brillinger}, \cite{lawley}
\begin{eqnarray}
\label{eq:E_li_sto} E(l_i) = \lambda_i + \sum_{j \neq i} \frac
{\lambda_i \lambda_j}{n(\lambda_i - \lambda_j)} + \mathcal
O(n^{-2})
\\
\label{eq:cov_li_lj_sto} \textrm{Cov}(l_i,l_j) =  \delta_{ij}
n^{-1} \lambda_i^2 + \mathcal O(n^{-2}).  \quad
\end{eqnarray}
\end{theorem}
where $\delta_{ij}$ is the Kronecker delta function. Now we
generalize Theorem \ref{th:Girshick} to the non-central case.
\begin{theorem}
\label{theorem} Let $\textrm{vec} (\boldsymbol X) \sim
\mathcal{N}(\textrm{vec} (\boldsymbol \mu) \, , \boldsymbol I_n
\otimes \, \sigma^2 \boldsymbol I_L)$. Then asymptotically for the
signal eigenvalues of $\hat {\boldsymbol R}$ we will have
\begin{eqnarray}
\label{eq:E_li_det} E(l_i) = \lambda_i + \sum_{j \neq i}\frac
{(\lambda_i + \lambda_j) \, \sigma^2 - \sigma^4}{ n(\lambda_i -
\lambda_j)} + \mathcal O(n^{-2})
\\
\label{eq:cov_li_lj_det} \textrm{Cov}(l_i,l_j) = \delta_{ij} \,
n^{-1} ( 2 \lambda_i \sigma^2 - \sigma^4) + \mathcal O(n^{-2})
\end{eqnarray}
\end{theorem}

\begin{proof}
See Appendix \ref{app:thrm}.
\end{proof}

\subsection{Noise Eigenvalues}
\label{subsec:stat_noise}

The eigenvalues associated with the noise subspace come from a
spherical subspace. Therefore, they are not sufficiently
separated, but placed tight together around the noise power
$\sigma^2$. Then, the perturbation analysis in Appendix
\ref{app:thrm} is no longer true, since their eigenvectors change
dramatically with a small perturbation in $\boldsymbol R$. The
distribution of the noise eigenvalues is identical to the
noise-only observations in an $L-d$ dimensional noise subspace
with a small negative bias introduced by signal eigenvalues
\cite{johnstone}. Here, we introduce two statistical distributions
to show that some noise eigenvalues are considerably larger than
$\sigma^2$. This invalidates the approximations used in
\cite{wang} for calculating $p_m$. In low SNRs, the weakest signal
eigenvalue approaches the largest noise eigenvalue but cannot pass
it due to the ordering of the eigenvalues. In this subsection, we assume $\sigma^2 = 1$.

\subsubsection{The Mar\v{c}enko-Pastur distribution}

For sufficiently large $n$ and $L$, with $\gamma = n/L$ and in the
null case, the distribution of unordered noise eigenvalues is
\cite {johnstone}
\begin{equation}
\label{eq:pastur} g(l) = \frac{\gamma}{2\pi l} \sqrt{(b-l)(l-a)}
\quad : \quad a \leq l \leq b
\end{equation}
where $a=(1-\gamma^{-1/2})^2$, $b=(1+\gamma^{-1/2})^2$, as
depicted in Fig. \ref{fig:pastur}. Note that $g(l)$ is a univariate distribution since it expresses the \emph{bulk} distribution \cite{johnstone} of the eigenvalues, i.e., in the null case, the eigenvalues of the covariance matrix are $L$ independent samples of this distribution. 
\begin{figure}
\centering
\includegraphics[width = 0.45 \textwidth]{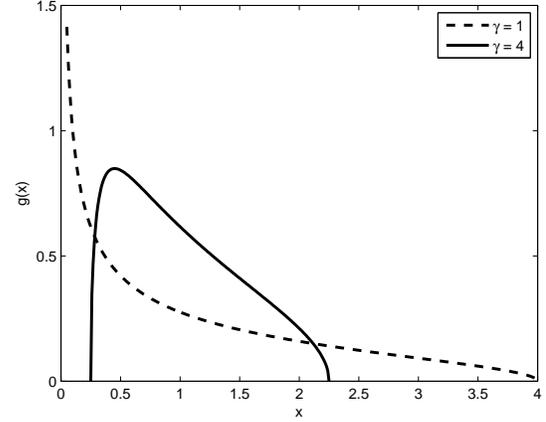}
\caption{Limiting densities of the noise subspace eigenvalues for
$\gamma=1$ and $\gamma=4$ cases. The spread of the eigenvalues
around 1 is evident.} \label{fig:pastur}
\end{figure}

\subsubsection{The Tracy-Widom distribution}

The largest eigenvalue of a complex correlation matrix in the null
case has a bell-shaped distribution called $F_2$ with moments
\cite {johnstone}
\begin{equation}
E(l_1) \simeq \mu_{nL} - 1.8 \, \sigma_{nL}
\end{equation}
\begin{equation}
\textrm{Std} (l_1) \simeq 0.9 \, \sigma_{nL}
\end{equation}
in which
\begin{equation}
\mu_{nL}= \bigg(1 + \sqrt{\frac{L}{n}} \, \bigg)^2
\end{equation}
\begin{equation}
\sigma_{nL}= \sqrt {\frac{\mu_{nL}}{n}} \bigg(\frac{1}{\sqrt{n}} +
\frac{1}{\sqrt L} \, \bigg)^{1/3}.
\end{equation}
Let's see a numerical example. Assume $n=100$ and $L=10$, then
$E(l_1) \simeq 1.55$ and $\textrm{Std} (l_1) \simeq 0.09$ which
implies that $l_1 > 1.3$ with high probability. We conclude that
the signal eigenvalues should be well larger than $\sigma^2$.

\section{Probability of Missed Detection}
\label{sec:p_miss}

\subsection{Method of Calculation}
\label{sec:p_miss_calc}

In this subsection, using the statistical tools developed in the
previous section, we calculate $p_m$ for MDL method. $p_{fa}$ is
negligible in moderate values of $n$ and $L$. For example, in
$L=3$ and $n=30$, $p_{fa} \simeq 0.003$ and decays rapidly when
$n$ and $L$ increase. $p_m$ can be used to estimate the minimum
energy level of a source to be detectable by the system. It can
also be used to determine the system capability for resolving very
close sources. Then, we concentrate on the $p_{m1} \triangleq
p_m(d=1)$ and $p_{m2} \triangleq p_m(d=2)$, although our method
can be used for the general scenario. Let $H_1$ denote the
situation in which only one source is present
\begin{equation}
\label{eq:pm1_def} p_{m1} = p \, \big(\Lambda (0,L,n) < \Lambda
(1,L,n) \, | \, H_1 \big).
\end{equation}
Using \eqref{eq:MDL} and rearranging the terms in
\eqref{eq:pm1_def} we get
\begin{eqnarray}
\label{eq:pm1_rearr} p_{m1} = p \, \big( L \log \bigg(
\frac{a_0}{\,g_0}
\bigg) - (L-1) \log \bigg( \frac{a_1}{g_1} \bigg) \qquad \nonumber \\
< \frac{1}{2n} (2L-1) \log(n) \big)
\end{eqnarray}
By the definition of $a_d$ in \eqref{eq:a_d}, we can write
\begin{equation}
\label{eq:a0_a1} a_0 = \frac{1}{L} \, l_1 + \frac{L-1}{L} \, a_1
\end{equation}
Similarly, for the geometric mean using \eqref{eq:g_d} we have
\begin{equation}
\label{eq:g0_g1} g_0^L = l_1 \, g_1^{L-1}
\end{equation}
Substituting \eqref{eq:a0_a1} and \eqref{eq:g0_g1} in
\eqref{eq:pm1_rearr}, we get \cite{wang}
\begin{equation}
\label{eq:pm1_Qm} p_{m1} = p \, \big( \log Q_{m1} \big(\textstyle
{\frac{l_1}{a_1}} \big) < T_1 \big)
\end{equation}
where
\begin{equation}
\label{eq:Qm1} Q_{m1}(x) \triangleq \frac{1}{x} \Big( \, 1 +
\frac{x-1}{L} \, \Big)^L
\end{equation}
and
\begin{equation}
\label{eq:T} T_1 = \frac{1}{2n} (2L-1) \log(n)
\end{equation}

In \cite{wang}, The function $\log (Q_{m1}(x))$ is approximated by
its second order Taylor series near $x=1$. This is one source of
avoidable error in the method. The smallest eigenvalue of the
signal subspace is greater than the largest eigenvalue of the
noise subspace, which is, from subsection \ref{subsec:stat_noise},
larger than $\sigma^2$. Also recall that $a_1 \simeq \sigma^2$, we
conclude that $x>1$. It is evident that the function $\log
(Q_{m1}(x))$ is uniformly increasing in the region $x>1$,
therefore we can translate the inequality in \eqref{eq:pm1_Qm} to
a simpler one
\begin{equation}
\label{eq:pm1_Tx} p_{m1} = p \, ( x < T_{1\,x} )
\end{equation}
where
\begin{equation}
\label{eq:T1x} \log(Q_{m1}(T_{1\,x})) = T_1
\end{equation}
Using \eqref{eq:pm1_Tx}, two steps are required for calculation of
$p_m$, computing $T_{1\,x}$ from \eqref{eq:T1x} and determining
the statistics of $x \triangleq {l_1}/{a_1}$ in \eqref{eq:pm1_Tx}.

Unfortunately, \eqref{eq:T1x} cannot be solved analytically for
$T_{1\,x}$, then we find an approximate solution in the first
step. Rearrange \eqref{eq:T1x} to get
\begin{equation}
\label{eq:Tx_rear} \bigg(1+ \frac{T_{1x}-1}{L} \bigg)^L = T_{1x}
\, e^{T_1}
\end{equation}
Expanding the left-hand-side of \eqref{eq:Tx_rear} to the second
order, assuming $L$ is sufficiently large and solving the
resulting quadratic equation, gives a first approximation for
$T_{1x}$
\begin{equation}
\label{eq:T1x_app} T_{1x}^{(1)} = 1+ \sqrt{e^{2T_1} - 1}
\end{equation}
Now since the function in L.H.S. of \eqref{eq:T1x} is smooth, we
can use a first order Taylor series around the solution in
\eqref{eq:T1x_app} to get closer to the exact solution
\begin{equation}
\label{eq:T1x_recur} T_{1x}^{(i+1)} = T_{1x}^{(i)} + \big( \, T_1
- T_1^{(i)} \big) \frac{T_{1x}^{(i)} + 1} {T_{1x}^{(i)} -1}
\end{equation}
where $T_1^{(i)}$ depends on $T_{1x}^{(i)}$ through
\eqref{eq:T1x}. Application of \eqref{eq:T1x_recur} for a few
times gives a very accurate solution. Note that computation of
$T_{1x}$ is done after setting $n$ and $L$, but is not dependent
on the SNR.

The next step in calculating $p_{m1}$ is determining the
statistics of $x$. From \eqref{eq:E_li_sto} and
\eqref{eq:cov_li_lj_sto}, we can see that $l_1$ is distributed as
\begin{equation}
\label{eq:dist_l1}  l_1  \sim \; \mathcal{N} \bigg(\lambda_1 +
\frac {(L-1)\lambda_1 \sigma^2}{n(\lambda_1 - \sigma^2 )} \;\, ,
\; \frac{\lambda_1^2} {n} \bigg)
\end{equation}
In \cite{wang, zhang, kavehCOR, fishler}, the bias term of $l_1$
is not considered, while a numerical example can clarify the
point. Assume that $n=100$, $L=10$, and $\sigma^2=1$. In the SNR
in which $p_{m1}$ starts to become large, $\lambda_1 = 1.5$,
$E(l_1)=2.2$, and $\textrm{Std} (l_1) = 0.15$. Therefore,
overlooking the bias term ($0.7$) introduces large error to the
analysis. Since in the critical SNRs, the signal eigenvalue get
closer to the noise eigenvalues, the denominator in
\eqref{eq:E_li_sto} reduces and the bias term gets large.

In the null case, $E(a_0)=\frac{1}{L}E( \textrm{Tr} (\hat
{\boldsymbol R})) = \sigma^2 = 1$, which recommends that $E(a_1 |
H_1) = \sigma^2$. But a signal eigenvalue can cause a negative
bias on $a_1$, numerically about 2\%. Then, although we neglect
the variance of $a_1$ which is very small compared to the variance
of $l_1$, we should take into account the bias to achieve an exact
performance evaluation. In fact, the variances of the eigenvalues (regardless of being a noise eigenvalue or a signal one) increases with the mean of the eigenvalue. This can be seen in the simulations and can be justified for the noise eigenvalues with noticing the decay of the Marcenko-Pastur distribution in Fig. \ref{fig:pastur} which results in increasing variance of its order statistics. The variance of any order statistic of a distribution is inversely proportional to the squared value of the distribution in the vicinity of the mean value of that order statistics. A classical example of this fact is the variance of the median. For the signal eigenvalues, this is already shown in \eqref{eq:cov_li_lj_sto} and \eqref{eq:cov_li_lj_det}. This fact, along with the averaging in the calculation of $a_1$ shows that its variance is negligible in the analysis. To calculate the bias, note that $E(l_1) + (L-1) E(a_1) = E(\textrm{Tr} (\hat{\boldsymbol R})) = \textrm{Tr} (\boldsymbol R) = \lambda_1 + (L-1) \sigma^2$. This besides \eqref{eq:E_li_sto}
gives \cite{wong}:
\begin{equation}
\label{eq:a(1)_bias} H_1 \quad : \quad a_1  \simeq  \; \sigma^2 -
\frac{\sigma^2 \lambda_1}{n ( \lambda_1 - \sigma^2)}
\end{equation}
Using \eqref{eq:dist_l1} and \eqref{eq:a(1)_bias}, the
distribution of $x$ is determined as a Gaussian random variable
with known mean $\mu_x$ and variance $\sigma_x^2$. Then, $p_{m1}$
can be calculated as
\begin{equation}
\label{eq:pm1} p_{m1} = 1 - Q \bigg( \frac{T_{1x} -
\mu_x}{\sigma_x} \bigg)
\end{equation}
in which 
\begin{equation}
Q(t) = \int _{t}^\infty \frac{1}{\sqrt {2\pi}} e ^{- \frac{u^2}{2}} du. 
\end{equation}
The same procedure can be used to calculate $p_{m2}$. The following approximation is widely used and justified in the literature \cite[eq. (24)]{wang}, \cite[eq. (II.3a)]{zhang}:
\begin{equation}
p_{m2} \simeq p \, \big(  \Lambda (1,L,n) < \Lambda (2,L,n) \, | \, H_2 \big)
\end{equation}
It basically states that the probability of missing one of the sources is very larger than missing both of them. We drop the details and just give some of the points important in the calculation of $p_{m2}$:
\begin{equation}
\label{eq:pm2_Qm} p_{m2} = p \, \big( \log Q_{m2} \big(\textstyle
{\frac{l_2}{a_2}} \big) < T_2 \big)
\end{equation}
in which the threshold $T_2$ and the function $Q_{m2}$ are defined as
\begin{equation}
\label{eq:T2} T_2 = \frac{1}{2n} (2L-3) \log (n)
\end{equation}
\begin{equation}
\label{eq:Qm2} Q_{m2}(x) = \frac{1}{x} \, \Big( \, 1 +
\frac{x-1}{L-1} \, \Big)^{L-1}
\end{equation}
\begin{equation}
\label{eq:x2}
x \triangleq \frac{l_2}{a_2}
\end{equation}
The recursive equation to estimate the threshold $T_{2x}$ will be
\begin{equation}
\label{eq:T2x_rec} T_{2x}^{(i+1)} = T_{2x}^{(i)} + \big( T_2 -
T_2^{(i)} \big) \frac{\; T_{2x}^{(i)} (L - 2 + T_{2x}^{(i)} ) }
{(L-2)(T_{2x}^{(i)} -1)}
\end{equation}
The distribution of $l_2$ will be 
\begin{equation}
\label{eq:dist_l2}  l_2  \sim \; \mathcal{N} \bigg(\lambda_2 +
\frac {(L-2)\lambda_2 \sigma^2}{n(\lambda_2 - \sigma^2 )} - \frac {\lambda_1 \lambda_2} {n(\lambda_1 - \lambda_2)}  \;\, ,
\; \frac{\lambda_2^2} {n} \bigg)
\end{equation}
$a_2$ will have a negligible variance and can be estimated by its mean value:
\begin{equation}
\label{eq:a(2)_bias} H_2 \quad : \quad E(a_2)  =  \; \sigma^2 -
\frac{\sigma^2 \lambda_1}{n ( \lambda_1 - \sigma^2)}  -  \frac{\sigma^2 \lambda_2}{n ( \lambda_2 - \sigma^2)}
\end{equation}
Now, using \eqref{eq:dist_l2} and \eqref{eq:a(2)_bias}, the distribution of $x$ in \eqref{eq:x2} can be found and $p_{m2}$ is achieved as in \eqref{eq:pm1}. The same procedure can be used for determining $p_m$ in any number of sources.

\subsection{Deterministic Signal Model}
\label{sec:det}

Although the first- and second-order statistical properties of the
signal subspace eigenvalues are different under stochastic and
deterministic signal models, the performance of the
MDL is the same under two models. As explained in section
\ref{sec:p_miss_calc}, $p_m$ depends on the statistics of the weakest
signal eigenvalue $l_d$. We show that these statistics grow
similar under two models when $l_d$ approaches the noise
eigenvalues.  Note that, for a fair comparison of the two signal models, the signal second-order characteristics should be the same (see e.g. \cite[sec. V]{ottersten}). Therefore, we have $ \lim _ {n \to \infty} \boldsymbol S_{\textrm{det}} \boldsymbol S_{\textrm{det}}^\textrm{\,H} / n =
E(\boldsymbol s_{\textrm{sto}} \boldsymbol
s_{\textrm{sto}}^\textrm{H})$, which results in $\boldsymbol
R_{\textrm{det}} = \boldsymbol R_{\textrm{sto}}$ and hence
$\lambda_{i \, \textrm{det}} = \lambda_{i \, \textrm{sto}}, \;
i=1, \ldots ,L$. In the situations where $p_m$ starts to grow
large, $l_d$ is barely larger than the noise eigenvalues,
$\lambda_d \simeq \sigma^{2}$, then from \eqref{eq:E_li_det} we
have
\begin{equation}
\label{eq:E_li_det_eq} E(l_{d \, \textrm{det}}) \simeq \lambda_d +
\sum_{i \neq d} \frac{\sigma^2 \lambda_i}{n(\sigma^2 - \lambda_i)}
\end{equation}
which is the same as \eqref{eq:E_li_sto} in stochastic signal
model. For the variances, we assume that $\lambda_d$ has
approached the upper limit of the noise eigenvalues
\begin{equation}
\label{eq:var_sto_det} \lambda_d \simeq \sigma^2
\bigg(1+\sqrt{\frac{L}{n}} \, \bigg)^2
\end{equation}
which is the upper limit of the Marcenko-pastur distribution in
\eqref{eq:pastur}. Note that, as signal power reduces, its
eigenvalue approaches the noise eigenvalues roughly about
$\sigma^2$. But $\lambda_d$ cannot be smaller than the largest
noise eigenvalue due to the sorting of the eigenvalues. Then as
the SNR reduces, $\lambda_d$ approaches the upper limit of the
noise eigenvalues about \eqref{eq:var_sto_det}. In fact, we are
using a better approximation for $\lambda_d$ in calculating the
variance in \eqref{eq:var_sto_det} rather than in calculating the
expectation in \eqref{eq:E_li_det_eq}. Assuming $L \ll n$, a first
order expansion of \eqref{eq:var_sto_det} can be used in
\eqref{eq:cov_li_lj_sto} to give
\begin{equation}
\label{eq:var_ld_sto} \textrm{Var}_\textrm{sto} \, (l_d) =
\frac{1}{n} \, \lambda_d ^2 \simeq \frac{1}{n} \, \sigma^4 \bigg(
1 + 4 \sqrt{\frac{L}{n}} \; \bigg)
\end{equation}
and in \eqref{eq:cov_li_lj_det} to give
\begin{eqnarray}
\label{eq:var_ld_det} \textrm{Var}_\textrm{det} \, (l_d) =
\frac{1}{n} \, \big( 2 \lambda_d \sigma^2 - \sigma^4 \big) \qquad
\qquad \qquad \qquad \nonumber \\ \simeq \frac{1}{n} \, \sigma^4
\bigg[ 2 \bigg( 1 + 2 \sqrt{\frac{L}{n}} \; \bigg) - 1 \bigg]
\end{eqnarray}
which reduces to the result in \eqref{eq:var_ld_sto} and we can
conclude that the variance of $l_d$ is the same under two models
in low SNRs. Hence, $p_m$ is approximately the same under two
signal models. This is in harmony with the same result in the DOA estimation problem, where the performance of the estimators are the same under two signal model \cite{ottersten}.

\section{Simulation Results}
\label{sec:simul}

In this section, simulation results are presented to support the
theoretical derivations. We consider $p_m$ in different conditions
of number of snapshots $n$, and number of sensors $L$ in a Uniform
Linear Array with half-wavelength inter-element distance. Our
estimate is compared with \cite{wang} and \cite{fishler}. Results
are presented for two closely spaced sources in $p_{m2}$, and one
source in $p_{m1}$. When the sources get closer to each other, the
weaker signal eigenvalue approaches the noise eigenvalues and
possibly miss will occur. Therefore, for a fixed angular distance
of the sources, a minimum SNR is required for the array to be able
to detect both sources.

Two equally powered uncorrelated signal sources in $\pm 2 ^o$ are
assumed. The SNR is defined as the ratio of each signal variance
to noise variance (i.e. sensor SNR).
\begin{figure}
\centering
\includegraphics[width = 0.5\textwidth]{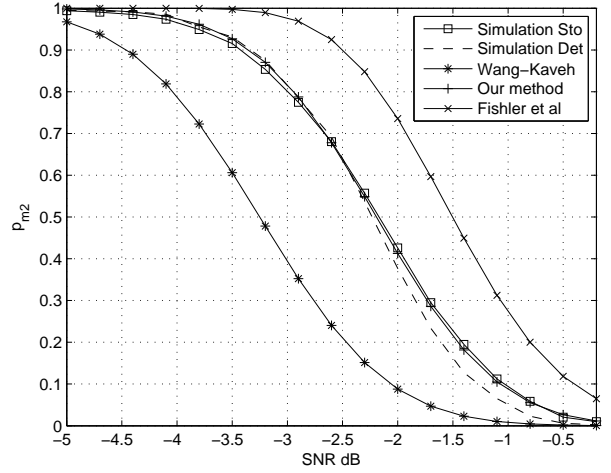}
\caption{$p_{m2}$ of MDL method when number of sensors $L=10$, and
number of snapshots $n=100$.} \label{fig:pm2_10_100}
\end{figure}
\begin{figure}
\centering
\includegraphics[width = 0.5\textwidth]{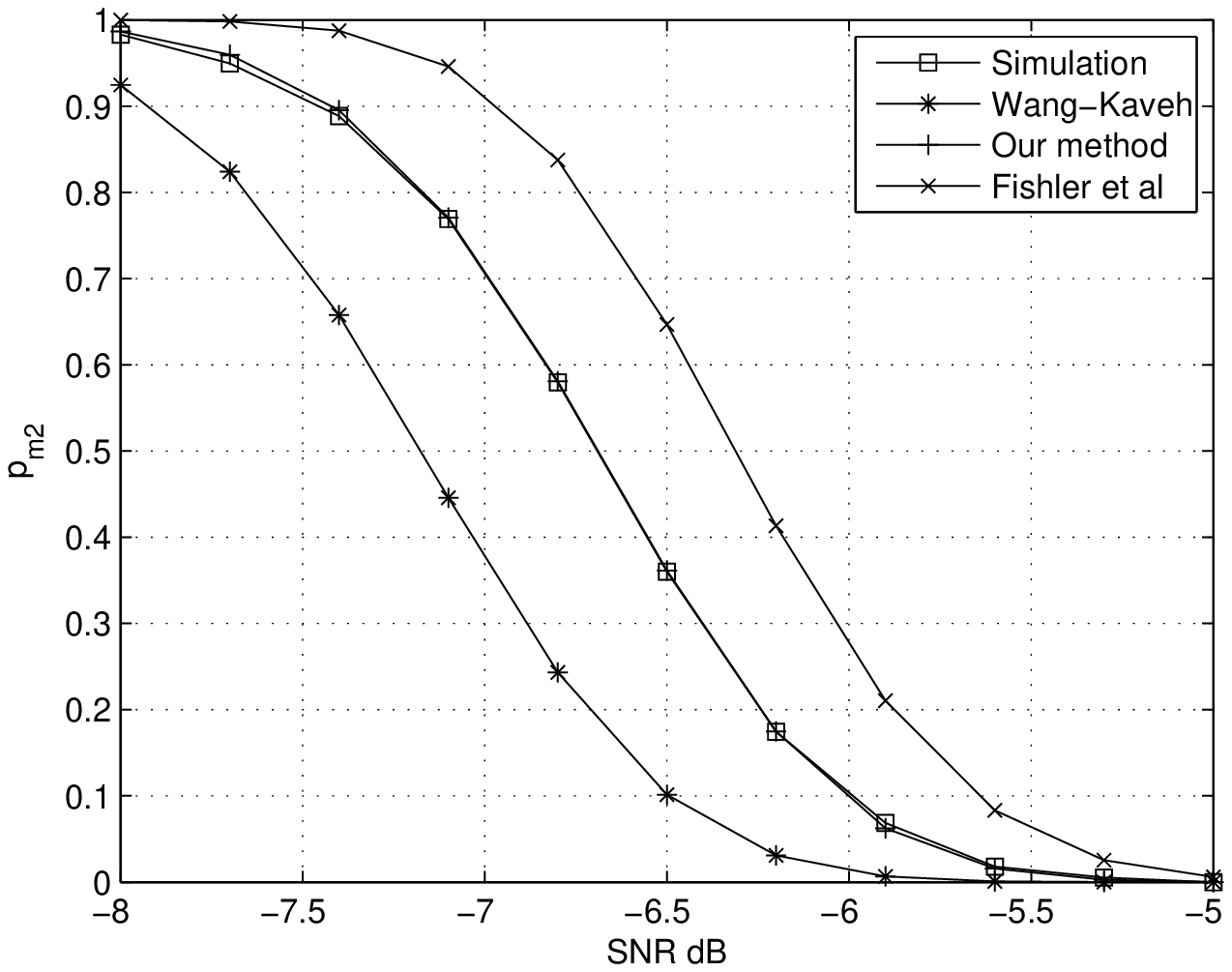}
\caption{$p_{m2}$ of MDL method when number of sensors $L=10$, and
number of snapshots $n=900$.} \label{fig:pm2_10_900}
\end{figure}
\begin{figure}
\centering
\includegraphics[width = 0.5\textwidth]{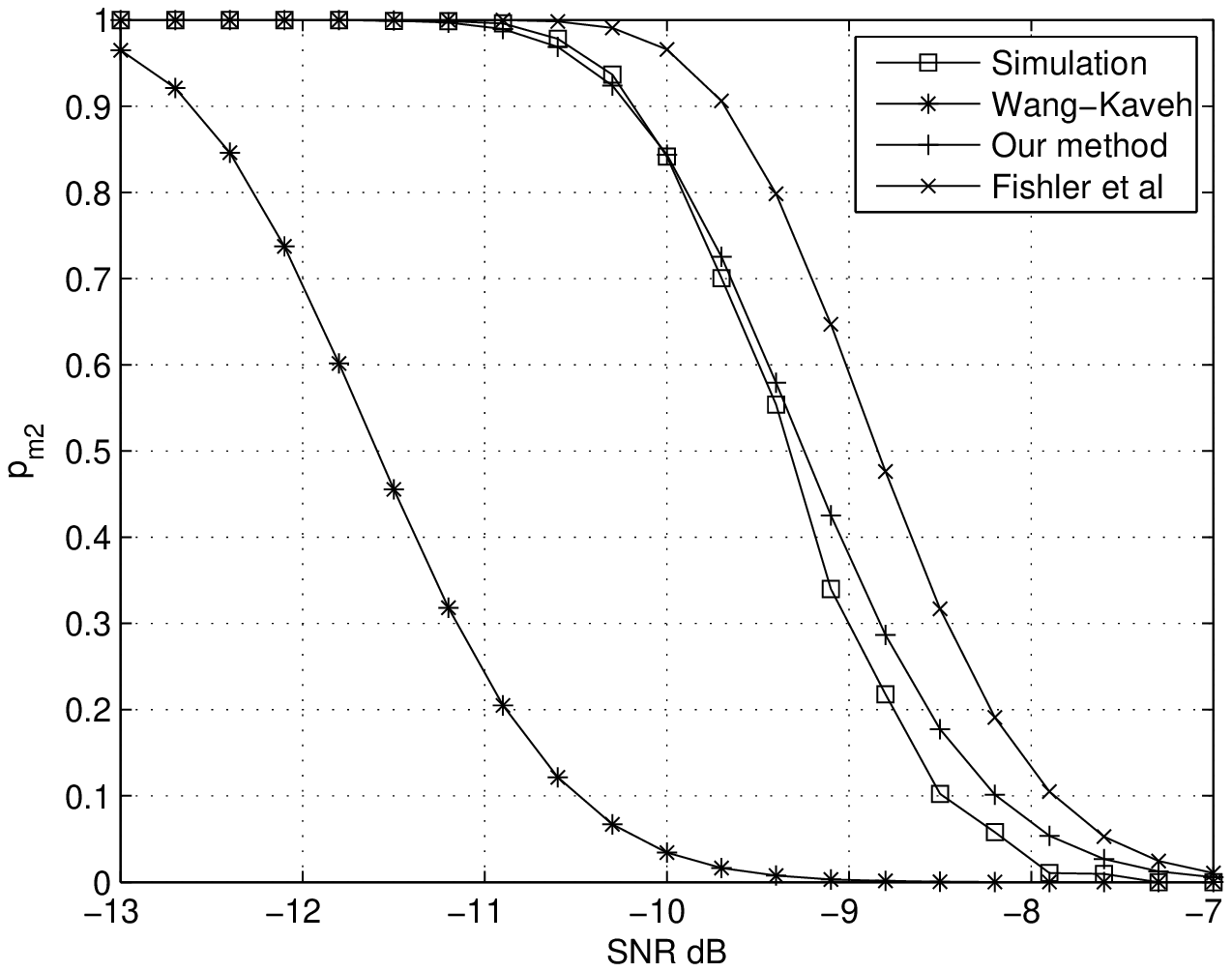}
\caption{$p_{m2}$ of MDL method when number of sensors $L=32$, and
number of snapshots $n=64$.} \label{fig:pm2_32_64}
\end{figure}
\begin{figure}
\centering
\includegraphics[width = 0.5\textwidth]{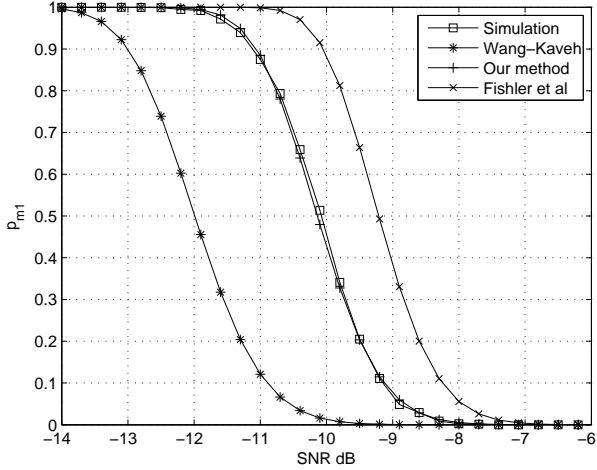}
\caption{$p_{m1}$ of MDL method when number of sensors $L=32$, and
number of snapshots $n=64$.} \label{fig:pm1_32_64}
\end{figure}
Figs \ref{fig:pm2_10_100}, \ref{fig:pm2_10_900}, and
\ref{fig:pm2_32_64} show the corresponding results for $p_{m2}$
different situations in terms of $n$ and $L$. Fig.
\ref{fig:pm1_32_64} presents the results for $p_{m1}$ in the worst
case of parameters. The superiority of our method in estimating
the simulation results is evident. In Fig. \ref{fig:pm2_10_100},
simulation results are presented for both deterministic and
stochastic signals, which confirms the approximate equality of
$p_m$ under two models. This equality improves as the number of
observations $n$ increases. Note that our method is used to
estimate $p_m$ under stochastic signal model in Fig. \ref
{fig:pm2_10_100}. The analysis in \cite{wang} under-estimates
$p_m$ with a horizontal distance of about 0.5-2 dB. In fact, this
method improves when $n$ gets larger since in this situation, the
neglected biases reduce. The estimate of \cite{fishler} is better
than \cite{wang}, with over-estimation of $p_m$ equivalent with a
horizontal distance about 0.5-1 dB. Note that in the extreme case
of $L=32$ and $n=64$ of Fig. \ref{fig:pm2_32_64}, our analysis
starts to degrade since the asymptotic assumption is no longer
valid. Though, in most cases, our estimate exhibits horizontal
distance of about 0.03 dB.

We have seen that the analysis in \cite{wang, kavehCOR, zhang}
lacks the inclusion of biases of the eigenvalues and also suffers
from some inaccurate approximations. But the analysis in
\cite{fishler} requires more scrutiny since as we have seen in the
simulation results, this analysis gives completely different
results from \cite{wang}. Authors in \cite{fishler} use asymptotic
conditions to show that $\Lambda(d-1) - \Lambda(d)$ converges in
distribution to a Gaussian random variable with mean $\mu$ and
variance $\sigma^2$. Simulations show that although the formula
derived for $\sigma^2$ in \cite{fishler} is a very good estimate
of the empirical value, the same is not true for the mean $\mu$,
which in fact shows considerable deviation. This disagreement is
present in small $n$ as well as large $n$ conditions. The derived
result for the mean of the Gaussian distribution in \cite[eq.
(19)]{fishler} is
\begin{eqnarray}
\label{eq:mu_fishler} \mu = n \log \bigg( \frac{\sigma_n^2}
{\lambda_d} \bigg[ 1+ \frac{1}{L-d+1} \Big( \, \frac{\lambda_d}
{\sigma_n^2} - 1\Big) \bigg] ^{L-d+1} \, \bigg) \nonumber \\
+ \; 0.5 \Big( 2d - 2L - 1 \Big) \log (n)
\end{eqnarray}
which we can see that is $n \log Q_{md}(x)$ plus some nonrandom
term in the notation of our analysis. Now, it is evident that
\eqref{eq:mu_fishler} is derived assuming $E(l_i) = \lambda_i$ for
signal subspace and $E(a_d) = \sigma_n^2$, thus every biases in
the distribution of $l_i$ and $a_d$ is ignored. Additionally,
Although we can assume the distribution of $x$ to be Gaussian, it
is not easy to assume normality for the function $\Lambda(d-1) -
\Lambda(d)$ since it is a highly nonlinear function of $x$.
Simulations show that the normality assumption is approximately
valid only for large values of $n$, say $n \simeq 1000$. Another
issue is that nonlinearity of the function $\log (Q_{md}(x))$ move
the mean of the distribution which is not taken into account. 

Here, we will give further simulation results that compare our analysis with the one presented in \cite{fishler}. We assume the same conditions as in \cite[Fig. 1]{fishler} which is $n=900$, $L=7$, and two Gaussian sources in $\boldsymbol \theta = [-5^o \; \, +10^o]$. 
\begin{figure}
\centering
\includegraphics[width = 0.5\textwidth]{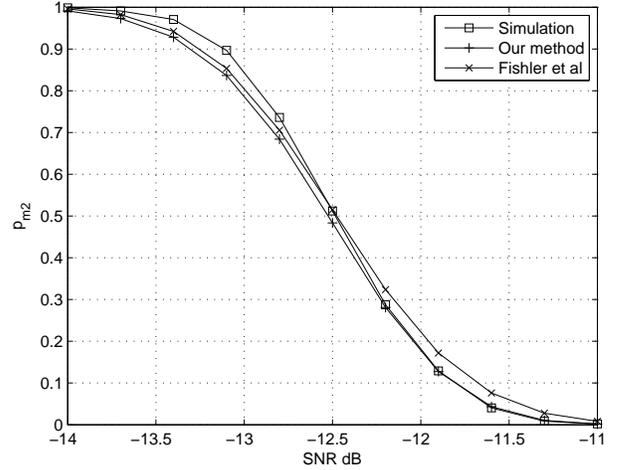}
\caption{$p_{m2}$ of MDL method when number of sensors $L=7$, and
number of snapshots $n=900$. The performance prediction method in \cite{fishler} works well in this set of parameters.} \label{fig:fishler1}
\end{figure}
The results are shown in Fig. \ref{fig:fishler1}, where the experimental performance of MDL method is accurately predicted by both our method and the method presented in \cite{fishler}. Although from a theoretical point of view, the method of \cite{fishler} is not comprehensive enough, in this special case of parameters it works well. If we change the sources DOAs and keep every other parameters unchanged we will see that the predictions of \cite{fishler} degrades. 
\begin{figure}
\centering
\includegraphics[width = 0.5\textwidth]{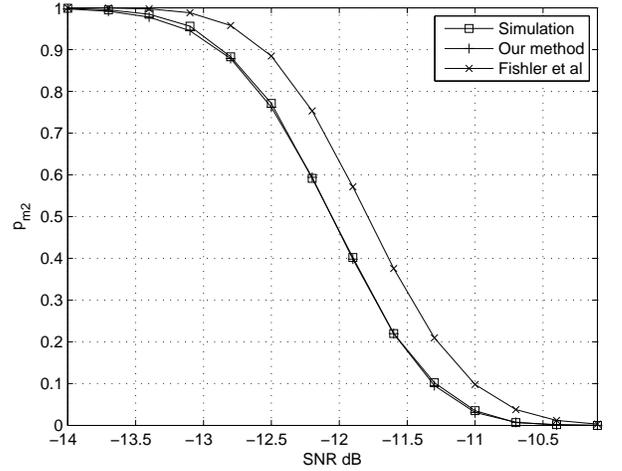}
\caption{$p_{m2}$ of MDL method when number of sensors $L=7$, and
number of snapshots $n=900$. The performance prediction method in \cite{fishler} does not work well in this set of parameters.} \label{fig:fishler2}
\end{figure}
Figure \ref{fig:fishler2} shows the experimental results and theoretical predictions when sources are in $\boldsymbol \theta = [-5^o \; \, 20^o]$. It is evident that the method of \cite{fishler} does not work well anymore while our method is still accurate. Note that we have investigated its performance when sources are very close to each other in our previous simulation results where the method in [8] failed to predict the performance accurately. Therefore, the method in [8] cannot be a reliable method of analytical performance calculation.

\section{Conclusion}
\label{sec:con}

An accurate performance analysis for the probability of missed
detection of the MDL source enumeration method was presented.
Statistical characterization of the principal components of the
covariance matrix helped to take good assumptions and
approximation which resulted in improved estimations of $p_m$. It
is proved that the performance is approximately identical under
stochastic and deterministic signal models using a perturbation
analysis which gives the statistical properties of eigenvalues in
the deterministic signal model. Simulation results show the
superiority of the proposed analysis compared with the previous
results.

\appendices

\section{Proof of Lemma \ref{lem:Bril_nc}}
\label{app:bril_nc}

Let $\boldsymbol X \triangleq \boldsymbol Y -
\boldsymbol \mu$ and rearrange the covariance in
(\ref{eq:bril_nc}) as
\begin{eqnarray}
\label{eq:covx} n^2c  =   \textrm{Cov} (\boldsymbol \alpha ^
\textrm{H} \boldsymbol {XX}^ \textrm{H} \boldsymbol \beta +
\boldsymbol \alpha ^ \textrm{H} \boldsymbol {\mu X}^ \textrm{H}
\boldsymbol \beta + \boldsymbol \alpha ^ \textrm{H} \boldsymbol
{X\mu}^ \textrm{H} \boldsymbol \beta \nonumber  \\  , \boldsymbol
\gamma ^ \textrm{H} \boldsymbol {XX}^ \textrm{H} \boldsymbol \zeta
+ \boldsymbol \gamma ^ \textrm{H} \boldsymbol {\mu X}^ \textrm{H}
\boldsymbol \zeta + \boldsymbol \gamma ^ \textrm{H} \boldsymbol {X
\mu}^ \textrm{H} \boldsymbol \zeta).
\end{eqnarray}
Circularity of the distribution and zero odd moments of zero-mean Gaussian
distribution reduces (\ref{eq:covx}) to
\begin{eqnarray}
\label{eq:covxsimp} n^2c  =   \textrm{Cov} (\boldsymbol \alpha ^
\textrm{H} \boldsymbol {XX}^ \textrm{H} \boldsymbol \beta \, ,
\boldsymbol \gamma ^ \textrm{H} \boldsymbol {XX}^ \textrm{H}
\boldsymbol \zeta) \nonumber \\ + \: \textrm{Cov} (\boldsymbol
\alpha ^ \textrm{H} \boldsymbol {\mu X}^ \textrm{H} \boldsymbol
\beta \, , \boldsymbol \gamma ^ \textrm{H} \boldsymbol {\mu X}^
\textrm{H} \boldsymbol \zeta) \nonumber \\ + \: \textrm{Cov}
(\boldsymbol \alpha ^ \textrm{H} \boldsymbol {X\mu}^ \textrm{H}
\boldsymbol \beta \, , \boldsymbol \gamma ^ \textrm{H} \boldsymbol
{X \mu}^ \textrm{H} \boldsymbol \zeta).
\end{eqnarray}
The first term in (\ref{eq:covxsimp}) is given by
(\ref{eq:brill}). The fact that $\boldsymbol x_i \perp \boldsymbol
x_j : i \ne j$ reduces the second term as
\begin{eqnarray}
\boldsymbol \alpha ^ \textrm{H} \! \boldsymbol \mu \,
E(\boldsymbol X^ \textrm{H} \boldsymbol {\beta \, \zeta} ^
\textrm{H} \boldsymbol X)\, \boldsymbol \mu ^ \textrm{H}
\boldsymbol \gamma = \nonumber \\ \boldsymbol \alpha ^ \textrm{H}
\! \boldsymbol \mu \, \textrm{diag} (E(\boldsymbol x_i^ \textrm{H}
\boldsymbol {\beta \, \zeta} ^ \textrm{H} \boldsymbol
x_i))\, \boldsymbol \mu ^ \textrm{H} \boldsymbol \gamma = \nonumber \\
(\boldsymbol \alpha ^ \textrm{H} \boldsymbol {\mu\mu}^ \textrm{H}
\, \boldsymbol \gamma )( \boldsymbol \zeta^ \textrm{H} \boldsymbol
\Sigma \, \boldsymbol \beta).
\end{eqnarray}
The third term in (\ref{eq:covxsimp}) can be derived in the same
way. Note that all the three terms in the right-hand-side of
\eqref{eq:bril_nc} are $\mathcal O({n^{-1}})$ since $\boldsymbol
\mu$ is of dimension $L \times n$ and hence $\boldsymbol {\mu\mu}^
\textrm{H}$ is $\mathcal O({n})$.

\section{Proof of Theorem \ref{theorem}}
\label{app:thrm}

In the asymptotic region of $n\gg1$, $\hat {\boldsymbol R}$ is a
slightly perturbed version of $\boldsymbol R$, described as
\begin{equation}
\label{eq:pertR} \hat {\boldsymbol R} = \boldsymbol R + p
\boldsymbol \Delta
\end{equation}
where $p \ll 1$ is the perturbation factor. Small perturbations in
$\boldsymbol R$ result in small changes in its eigenvectors if the
associated eigenvalues are sufficiently separated  \cite{golub}.
It means that the following results are true for signal
eigenvalues. Remember the definition of the eigendecompositions as
$\boldsymbol {Rv}_i = \lambda_i \boldsymbol v_i$ and
$\hat{\boldsymbol R} \boldsymbol w_i = l_i \boldsymbol w_i$. The
first order perturbation in eigenvectors is
\begin{equation}
\label{eq:pertw} \boldsymbol w_i \simeq \boldsymbol v_i + \sum_{j
\neq i} t_{ij}p  \, \boldsymbol v_j
\end{equation}
where $t_{ij}$s are the perturbation coefficients. Straightforward
calculations will give \cite[eq. (A.9)]{kaveh-barabell}\cite{wilkinson}:
\begin{eqnarray}
\label{eq:pert_li} l_i = \lambda_i + p \, \boldsymbol v_i^
\textrm{H} \! \boldsymbol {\Delta \, v}_i + \sum_{j \neq i}t_{ij}
p^2 \boldsymbol v_i^ \textrm{H} \boldsymbol {\Delta \, v}_j
\\
\label{eq:t1k} t_{ij} = \frac{\boldsymbol v_j^ \textrm{H}
\boldsymbol {\Delta \, v}_i}{\lambda_i - \lambda_j}. \qquad \qquad
\qquad
\end{eqnarray}
Under the conditions of Theorem \ref{theorem}, we will have
\begin{equation}
\label{eq:t_1k_2_det} \textrm{Cov} \big( \, t_{ik} , t_{jr} \big)
= \delta_{ij} \, \delta_{kr} \frac {(\lambda_i + \lambda_k) \,
\sigma^2 - \sigma^4}{ np^2(\lambda_i - \lambda_k)^2} \qquad
\end{equation}
which is shown using \eqref{eq:t1k} and replacing $\boldsymbol
{\mu \mu}^ \textrm{H} = n(\boldsymbol R - \sigma^2 \boldsymbol I)$
in (\ref{eq:bril_nc}). Now, \eqref{eq:E_li_det} is proved using
\eqref{eq:pert_li} and \eqref{eq:bril_nc}. \eqref
{eq:cov_li_lj_det} can be shown using \eqref{eq:pert_li} to the
first order and \eqref{eq:bril_nc}. Note that the limiting
distribution of the eigenvalues is Gaussian \cite{ottersten}.

\end{document}